\documentclass{article}
\usepackage{graphicx}  
\usepackage{dcolumn}   
\usepackage{bm}        
\usepackage{amssymb}   
\usepackage{amsmath}  
\usepackage{amscd,amsthm}
\usepackage{amsfonts}
\usepackage{mathtools} 
\usepackage{amssymb}
\usepackage{bm}
\usepackage{bbm}

\newtheorem{Satz}{Theorem}[section]

\newtheorem{Prop}[Satz]{Proposition}

\begin{document}

\title{Meta-universality classes at criticality}

\author{Duncan A.J.~Blythe\footnote{Berlin Institute of Technology, Berlin, Germany, Bernstein Centre for Computational Neuroscience, Berlin, Germany, African Institute of Mathematical Sciences, Bagamoyo, Tanzania}, Vadim V.~Nikulin\footnote{Charit{\'e} Medical University, Berlin, Germany, Bernstein Centre for Computational Neuroscience, Berlin, Germany, National Research University, Higher School of Economics, Moscow, Russia}}

\maketitle

\begin{abstract}
We show theoretically that the hypothesis of criticality as a theory of long-range fluctuation in the human brain may be distinguished from competing explanations on the basis of macroscopic neuronal 
signals such as the EEG, using novel theory of narrowband amplitude time-series at criticality. Our theory predicts the division of critical activity into \emph{meta-universality classes}.
As a consequence we provide strong evidence for criticality in the human brain on the basis of the EEG.
\end{abstract}

Testing for the presence of \emph{critical power-law avalanche dynamics} (CPLAD) is pivotal in establishing that a dynamical system is poised close to a critical point.
Such an analysis assumes that periods of activity are separated from periods of inactivity. For the intact brain of awake human subjects, which has been hypothesized to be critical \cite{beggs2008criticality}, the assumption of separation between activity and inactivity is never met, since the system is continuously active. This lack of separation makes unclear whether the results of existing experiments confirming criticality in animals \emph{in vitro} \cite{beggs2003neuronal,friedman2012universal} and \emph{in vivo} \cite{petermann2009spontaneous, ribeiro2010spike} generalize to the intact brain of awake human subjects; in particular it has been argued that the scale-free form of the power-spectra of LFP, EEG and MEG data is due not to CPLAD but to \emph{passive filtering} (PF) through the extracellular media \cite{bedard2006does,bedard2009macroscopic}. While several papers have claimed to test the hypothesis of criticality from large scale-brain signals of awake human subjects \cite{shriki2013neuronal} the authors only find CPLAD over short time-scales ($<1$ second) since the continuous nature of the data prevented testing of criticality as a theory of neuronal fluctuations over longer time-scales.
Quantifying potential CPLAD from such non-invasive human recordings is 
challenging since a superposition of overlapping avalanches may lead to incorrect estimation
of size and duration (critical) exponents (this is related to the under sampling problem \cite{ribeiro2010spike}) or indeed of whether
CPLAD are present at all. 
Interestingly, in this letter we show that in order to quantify and distinguish CPLAD from PF, it is not necessary to define individual avalanches; instead we show that is is sufficient to analyse the properties of continuous neuronal recordings. This approach allows us to test the hypothesis of criticality as a theory of neuronal variability over large time-scales. 

%
We use the facts that periods of activity of local neural networks are separated by periods of inactivity, and macroscopic brain signals, such as the LFP, EEG and MEG, measure the linear superposition $X'(t)$ of activity of numerous local neural networks \cite{griffiths1999introduction}; this means that the macroscopic brain signal is continuous, although the activity of local networks is not. Thus:
\begin{align}
\label{eq:1}
X'(t) = \sum_{s=1}^T \sum_{i=1}^{q_s} h_{s,i} a_{s,i} \left( \frac{t-s}{L_{s,i}}\right)
\end{align}
$a_{s,i} \left( \frac{t-s}{L_{s,i}}\right)$ denotes activity at time $t$ of the $i^\text{th}$ local neural network $L_{s,i}$
which begins at time $s$, after a period of inactivity and lasts for $L_{s,i}$ time steps. We adopt the conventions that $a_{s,i}(t)$ is normalized to unit standard deviation and is $0$
for $t \in (-\infty,0) \cup (1,\infty)$. Thus, $h_{s,i}$ denotes the average height of the time-course of activity of the $i^\text{th}$ local neural network which begins at time $s$.

PF and CPLAD differ on how $h_{s,i}$ and $L_{s,i}$ are distributed. According to the PF hypothesis their distributions decay faster than power-laws. For simplicity, therefore, we summarize the PF as claiming exponential distributions:
$p(L) \sim e^{-AL }$ and
$p(h) \sim e^{-Bh}$.
Moreover the PF hypothesis states power-spectra of macroscopic brain signals $X(t)$ are scale-free because they reflect a filtered version of $X'(t)$:
\begin{align}
\label{eq:2}
X(t) = \sum_{u=0}^\infty F(u) X'(t-u)
\end{align}
$F(u)$ is a linear filter due to the extracellular media which yields a signal with power-law power spectrum from white noise input \cite{bedard2006does,bedard2009macroscopic}.

On the other hand, the CPLAD hypothesis states that macroscopic signals $X(t)$ reflect $X'(t)$ directly so that $X(t) = X'(t)$ and that we have power law distributions for $h_{s,i}$ and $L_{s,i}$:
$L \sim L^{-\alpha} \nonumber$ and
$h \sim L^\beta$.
These power-law distributions explain the power-law form of the power-spectrum \cite{jensen19891}. In addition the CPLAD hypothesis states that each $a_{s,i}(t)$ is an independent and identical sample from a single stochastic process, which we call $a(t)$ \cite{mehta2002universal}.

We now present theory which makes predictions for the CPLAD hypothesis which provably do not hold for the PF hypothesis. As well as considering the raw signal X(t) we also consider the amplitude of a narrowband filtered version of $X(t)$, which we denote $g_\omega(X (t))$. Thus let $f_\omega(\cdot)$ be a linear narrowband filter with pass band 
$\omega \in [\omega - \Delta \omega],  \omega+ \Delta \omega]$, and $\mathcal{H}(\cdot)$ the Hilbert transform then:
\begin{align*}
g_\omega(X (t)) = |f_\omega(X (t)) + i\mathcal{H}(f_\omega(X (t)))| 
\end{align*}

The theory depends on two measures of long-term variation evaluated on $X(t)$ and $g_\omega(X (t))$: the Hurst exponent and the detrended cross correlation coefficient. The Hurst exponent $H$, of a stationary process $Y(t)$ may be defined by the scaling of the auto covariance function. For $H \in (0, 1)$, $Y(t)$ is said to be long-range temporally correlated (LRTC):
\begin{align*}
\mathbb{E}(Y(t + s)Y(t)) - \mathbb{E}(Y (t))^2 \sim  s^{2H-2}
\end{align*}
For auto covariances decaying faster than $s^{-1}$, one defines $H = 1/2$ and $Y(t)$ is not LRTC. From now on we distinguish the Hurst exponents of $X(t)$ and $g_\omega(X(t))$ as $H_{raw}$ and $H^\omega_{amp}$.
Typically in neuroscientific applications Hurst exponents are measured with Detrended Fluctuation Analysis (DFA) \cite{peng1994mosaic}.

The detrended cross correlation coefficient \cite{zebende2011study}
$\rho_{DCCA}(n, Y_1, Y_2)$ is a measure of correlation between two time-series $Y_1(t)$ and $Y_2(t)$ at a time-scale $n$,
which is invariant to non-stationary trends of a fixed
polynomial degree. Let $F^2_{Y_1,Y_1}(n)$ be the detrended variance of $Y_1$ where each window of size $n$ is detrended, as computed for DFA and $F^2_{Y_1,Y_2}(n)$, 
by analogy, the detrended covariance, as computed for DCCA \cite{podobnik2008detrended}; then in analogy to the Pearson correlation coefficient, one defines:
\begin{align*}
\rho_{DCCA}(n,Y_1,Y_2) = \frac{F^2_{Y_1,Y_2}(n)}{\sqrt{F^2_{Y_1,Y_1}(n) F^2_{Y_2,Y_2}(n)}}
\end{align*}
Our first result is that for the $X(t)$ of the PF hypothesis, Equation~\eqref{eq:2}, $H^\omega_{amp} = 0.5$ and, when $\omega_1 \neq \omega_2$, $\rho_{DCCA}(n, g_{\omega_1}(X(t)), g(\omega_2)(X(t))) \rightarrow 0$ as $n\rightarrow\infty$.
This can be seen by splitting Equation~\eqref{eq:1} into activity which lasts longer than some value $L'$ and activity which has duration shorter than or equal to $L'$:
\begin{align}
X'(t) = \sum_{L_{s,i}\leq L'} h_{s,i} a_{s,i}\left(\frac{t-s}{L_{s,i}}\right) + \sum_{L_{s,i}>L'} h_{s,i} a_{s,i} \left(\frac{t-s}{L_{s,i}} \right)
\end{align}
Since the distribution of L decays exponentially, the vari-
ance of the left hand term dominates. Therefore:
\begin{align*}
f_\omega(X'(t)) \sim \sum_{L_{s,i} \leq L'} h_{s,i} f_\omega\left(a_{s,i}\left(\frac{t-s}{L_{s,i}}\right)\right)
\end{align*}
Since time-points in this expression spaced more than
$L'$ points apart are independent we also have that
time-points of $g_\omega(X'(t))$ spaced more than $L'$ points
apart are independent, so that $H^\omega_{amp}= 1/2$ and 
$\rho_{DCCA}(n,\omega_1,\omega_2) \rightarrow \infty$ as $n\rightarrow \infty$ (same asymptotic properties as white noise); the reweighting of frequency induced by the passive filtering does not change the narrowband properties. See Proposition 2.1 of the supplement.

We now consider the CPLAD hypothesis. The behaviours we derive for $H^\omega_{amp}$ and
$\rho_{DCCA}(n,g_{\omega_1}(X(t)),g_{\omega_2}(X(t)))$  depend on the
exponents $\alpha$ and $\beta$. We find that there are four regions of parameters with qualitatively differing behaviours which we term  \emph{meta-universality classes}.
We assume that all power-law distributions are cut off at a lifetime $L_c$ which is proportional to the size of the system and, for simplicity, that at each time-point a fixed number $q_s = q$ of avalanches begin.

\noindent
({\bf MU1}) $\alpha < 2$

For $\alpha < 2$ the number of avalanches active at time $t$ is:
\begin{align*}
\sum_{s = 0}^{L_c} \#\{L_{t-s,i} | L_{t-s,i}>s\} &\sim \int_0^{L_c} q \left( \int_s^{L_c} L^{-\alpha} dL\right) \\
&\sim L_c^{2-\alpha}
\end{align*}
This implies the number of avalanches active at any given time is unbounded in the system size. Applying the Central Limit Theorem, this implies that $X(t)$ is a Gaussian process with power-law autocorrelation (see supplement for details). Since Gaussian processes are uniquely defined by their second order properties and may be generated by filtering white noise \cite{amblard2013basic}, then in analogy to the results for the PF hypothesis $H^\omega_{amp} = 1/2$ and $\rho_{DCCA}(n,\omega_1,\omega_2) \rightarrow 0$ as $n \rightarrow \infty$. See Proposition 2.4.

\noindent
({\bf MU2}) $\alpha >2$ and $\alpha< \beta+3$

For $\alpha > 2$, the probability that an avalanche is active with duration greater than $L'$ is:
\begin{align*}
\frac{1}{T} \sum_{L_{s,I}>L'} L_{s,i} \sim L'^{2-\alpha}
\end{align*}
Thus the probability that large avalanches occur simultaneously is negligible.
Moreover we have, for an exponent $\beta'$, a scaling relation for large $L$:
\begin{align*}
f_\omega(L^\beta a(t/L)) \sim L^{\beta'} f_\omega(a)(t/L) 
\end{align*}
Here $f_\omega(a)(t/L)$ is understood as position $t/L$ of $a(t)$ filtered in the narrowband around $\omega$. 
We are able to derive this exponent theoretically (Proposition 2.7) and find that $\beta' = \beta/2$.
This implies that for $\alpha<\beta+3$ (but not otherwise) the variance of the large narrowband filtered avalanches dominates so that:
\begin{align*}
f_\omega(X(t)) \sim \sum_{L_{s,i}>L'} L_{s,i}^{\beta/2} f_\omega(a) \left(\frac{t-s}{L_{s,i}}\right)
\end{align*}
Since the avalanches on the right hand side of this relation do not overlap then we have that:
\begin{align*}
g_\omega(X(t)) \sim \sum_{L_{s,i}>L'} L_{s,i}^{\beta/2} g_\omega(a) \left(\frac{t-s}{L_{s,i}}\right)
\end{align*}
Given this representation of the amplitude as a simple sum of amplitudes of individual avalanches, standard techniques may then be applied to approximating the Hurst exponent $H^\omega_{amp}$ \cite{jensen19891}, and we find:
\begin{align*}
H^\omega_{amp} &= \beta/2 -\alpha/2  + 2 \\
&> 1/2
\end{align*}
Similarly we find:
\begin{align*}
H^\omega_{raw} &= \beta -\alpha/2  + 2 \\
&> 1/2
\end{align*}
Moreover, assuming that the integrals $\int_0^1 g_\omega(a(t))dt$ exist, 
then the separation of large avalanches make it simple to derive that
$\rho_{DCCA}(n,\omega_1,\omega_2)\rightarrow 1$ as $n\rightarrow \infty$.
See Propositions 2.8 and 2.9.

\noindent
({\bf MU3}) $\alpha >2$ and $\alpha<2\beta+3$

For frequencies with $1/\omega \gg L_c$ and $L \leq L_c$:
\begin{align*}
f_\omega(L^\beta) \sim L^\beta f_\omega(a)(t/L)
\end{align*}
This is because relative to the time-scale of the filter, $a(t/L)$ may be treated as a delta function, which weights all frequencies equally. Therefore applying a similar argument as for {\bf MU2} we find that for $\omega_1$ , $\omega_2 \rightarrow 0$, 
$\rho_{DCCA}(n,\omega_1,\omega_2)\rightarrow 1$ as $n\rightarrow \infty$. Moreover, applying the results for {\bf MU2} we have:
\begin{align*}
H^\omega_{amp} &= 1/2 \\
H^\omega_{raw} &= \beta -\alpha/2  + 2 \\
&> 1/2
\end{align*}
See Proposition 2.10.

\noindent
({\bf MU4}) $\alpha > 2 \beta +3 $ and $\alpha > 2$

Since these universality classes have the shortest tails in
their critical distributions, we may apply identical methods as
applied to the PF hypothesis which show that $H^\omega_{amp} = 1/2$
 and $\rho_{DCCA}(n,\omega_1,\omega_2)\rightarrow 0$ as $n\rightarrow \infty$ when $\omega_1 \neq \omega_2$.
 See Proposition 2.11.

The results of the theory for the CPLAD model are summarized in Figure~\ref{fig:phase_space}.

\begin{figure}
\includegraphics[width = 0.8\columnwidth]{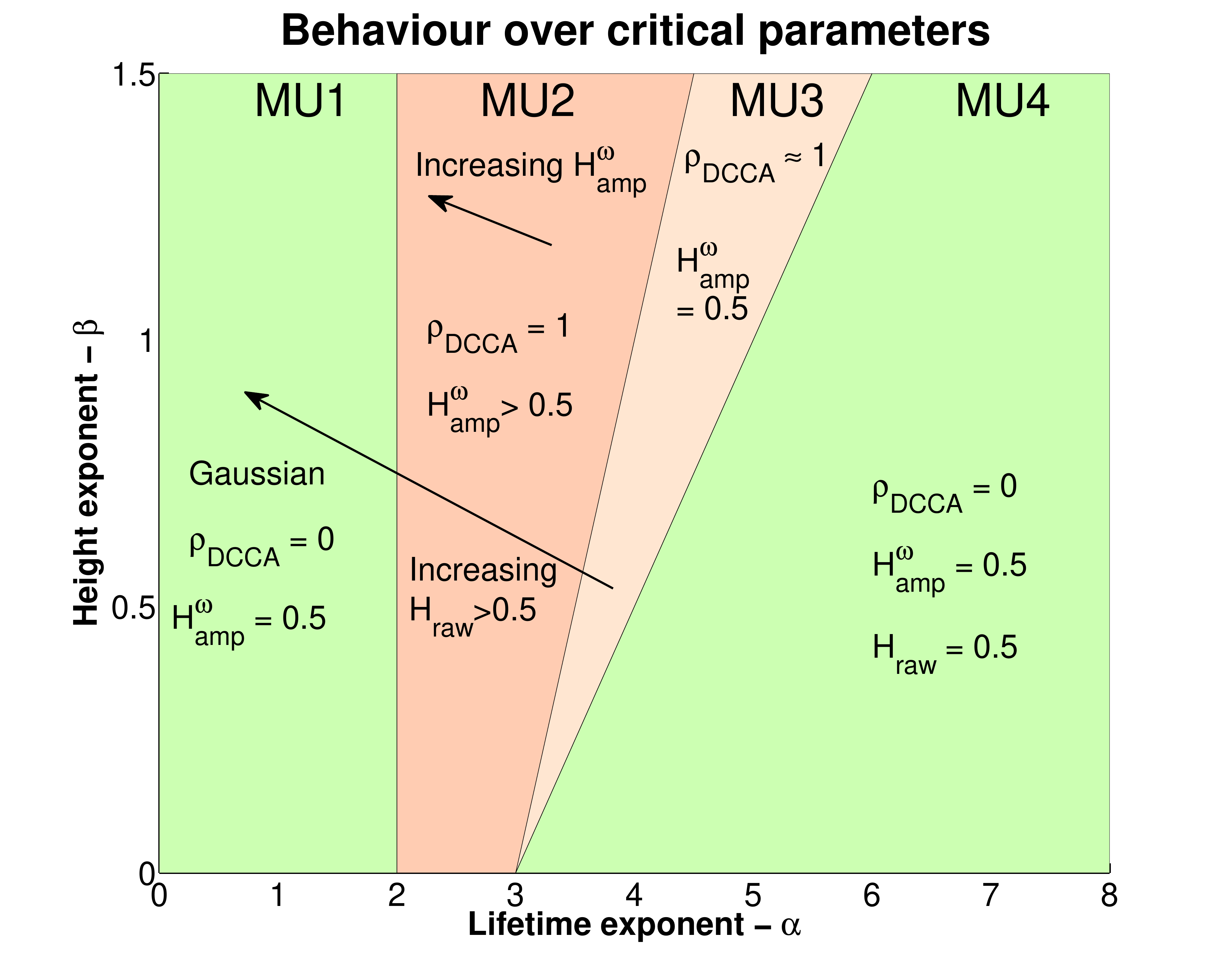}
\caption{Division of critical exponents $\alpha$ and $\beta$ into meta- universality classes. The figure displays the range of qualitative behaviours we predict with our theory. Areas marked in green display no LRTC behaviour in sub-bands or DCCA correlations between sub-bands. Areas in red display LRTC and/or cross correlations between amplitudes of sub-bands ($H^\omega_{amp} = 1/2$, $\rho_{DCCA}(n) = 0$ for large $n$).}
\label{fig:phase_space}
\end{figure} 

We now present tests of these predictions in simulations.
We first tested the predictions for the PF hypothesis by generating a LRTC process by linear filtering of white noise. 
The results are displayed in the bottom panel of Figure~\ref{fig:slow_example} and Figure~2 of the supplement.
The results confirm that $H^\omega_{amp} = 1/2$ and $\rho_{DCCA}(n,\omega_1,\omega_2)\rightarrow 0$.

\begin{figure}
\includegraphics[width = 0.8\columnwidth]{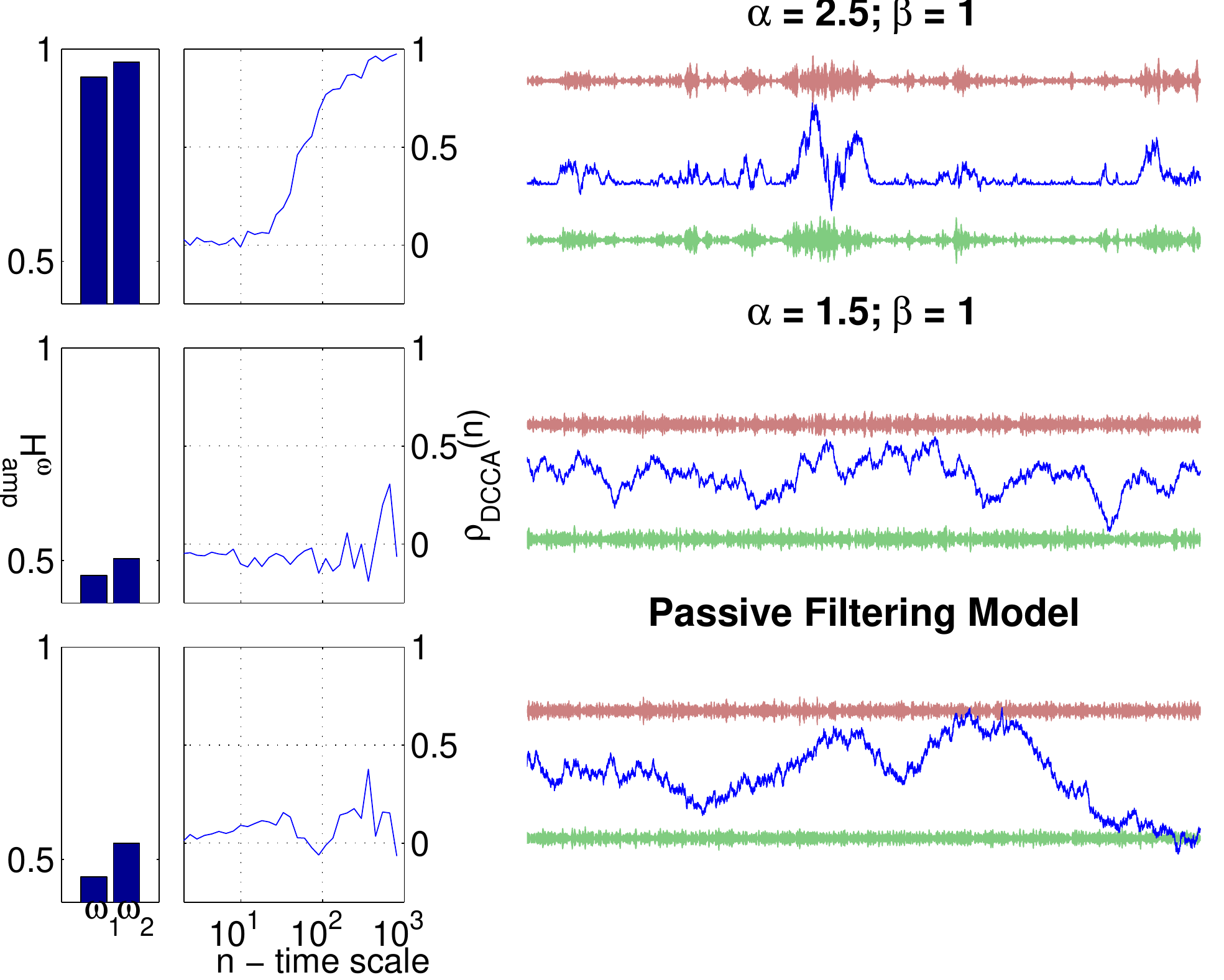}
\caption{Right: Sample paths from the PF and CPLAD models. 
In each of the three cases the $x$-axis denotes time 
and the $y$-axis number of activations.
In each case the middle trace denotes $X(t)$ and the top and bottom $g_{\omega_i}(X(t))$. 
Left: DCCA correlation coefficients $\rho_{DCCA}(n)$ and Hurst exponents $H^\omega_{amp}$.
}
\label{fig:slow_example}
\end{figure}

We then tested the predictions for the CPLAD hypothesis, modeling the activity of local networks by:
\begin{align*}
a(t) = b(t) + c(t)\epsilon(t)
\end{align*}
$b(t)$ is the average avalanche shape, $c(t)$ the shape in the variance profile and $\epsilon(t)$ a colored noise with the spectrum known for a critical system 
$\mathcal{P}[\omega] \sim \omega^{-\beta-1}$ \cite{kuntz2000noise} (see supplement for details). Two examples in {\bf MU2} and {\bf MU1} are displayed in the top two rows of Figure~\ref{fig:slow_example}.

The results of a simulation for all critical parameters are displayed in Figure~\ref{fig:cutoff_corr}. We find good agreement be- tween the meta-universality classes derived and the simulation results. (See supplement sec. 3.3 for details.)

\begin{figure}
\includegraphics[width=\columnwidth]{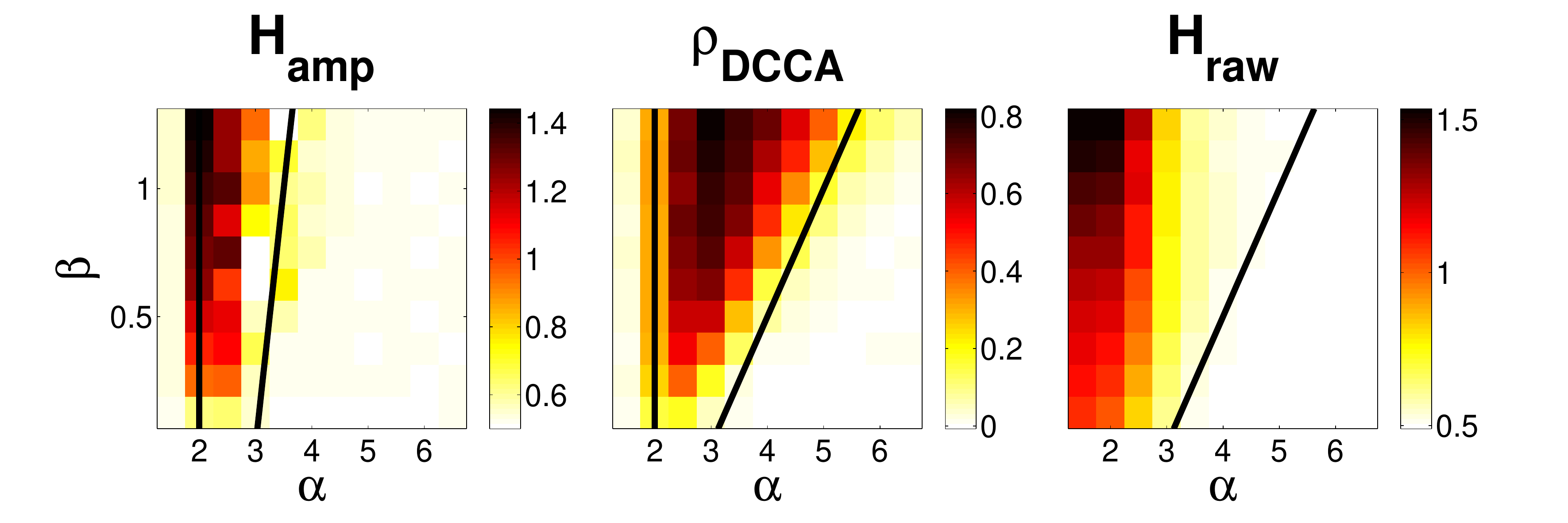}
\caption{DCCA correlation coefficients and Hurst exponent for the simulated CPLAD model.
The $x$-axis denotes the exponent $\alpha$ 
and the $y$-axis denotes $\beta$. 
The black lines denotes the transitions in meta-universality class according to the theory. 
\label{fig:cutoff_corr}
}
\end{figure}

Finally we tested the PF and CPLAD hypotheses by estimating $H^\omega_{amp}$, $H_{raw}$ and
$\rho_{DCCA}(n,\omega_1,\omega_2)$ on EEG time-series (see supplement for preprocessing).
Important is that we analyse 3 frequency ranges 
without oscillations (no local maximum in
power-spectrum); the aim was to restrict
analysis to activity corresponding to the $1/f^\gamma$ shape of
the power-spectra. Given that the data were sampled at
200Hz, and that lower frequencies require far larger
window sizes for analysis, we chose 3 frequencies above
the beta range, taking care to exclude the 50Hz line
noise.  

\begin{figure}
\includegraphics[width = 1\columnwidth]{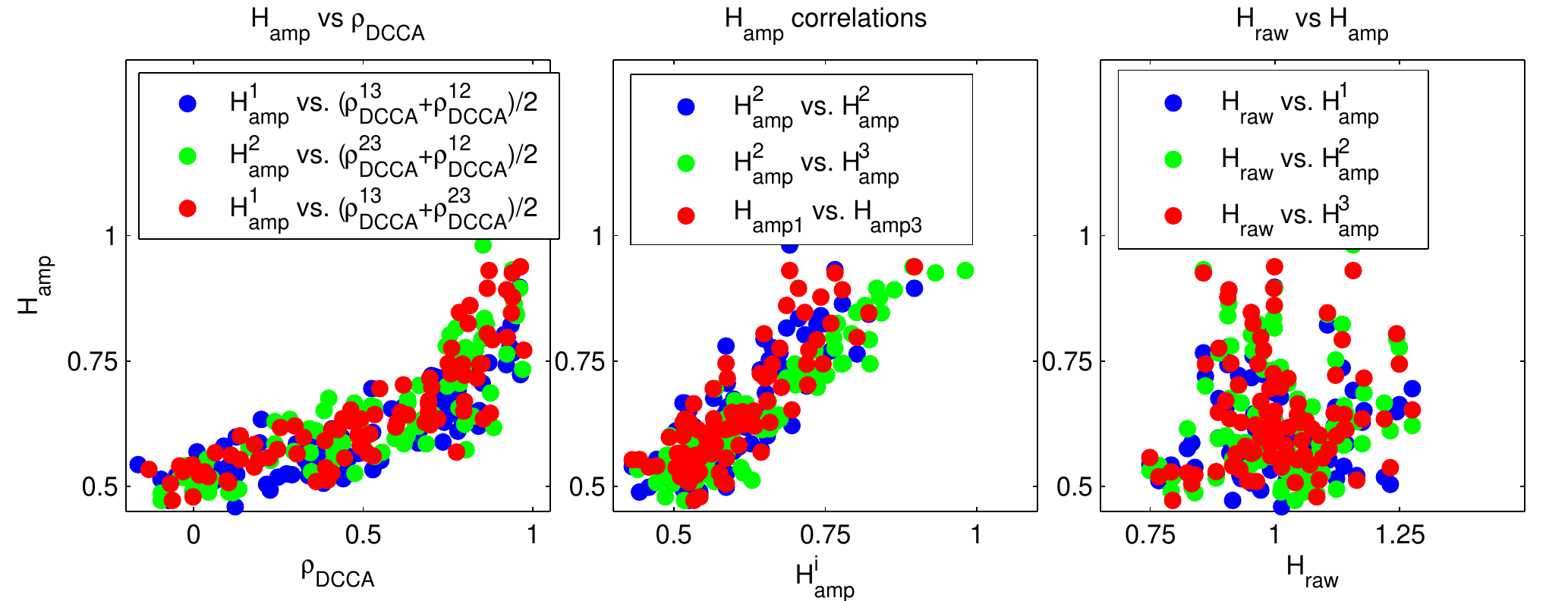}
\caption{Results of data analysis of human EEG. The frequency ranges analysed $i = 1,2,3$ are 35-40Hz, 60-65Hz and 72-77 Hz respectively, which are displayed as superscripts in the plots. }
\label{fig:real_data_prl}
\end{figure}

The results of our analysis over
all 7 subjects are displayed in Figure~\ref{fig:real_data_prl}. For 244 of 261
signals, the
DFA estimate of $H^\omega_{amp}$ was higher than $1/2$, clearly
indicating the presence of LRTC (median $H^\omega_{amp}$ = 0.61, 5th and 95th percentiles 0.49 and 0.85, p < 0.0001). Likewise, 238 of 261 DCCA correlation coefficients $\rho_{DCCA}(n,\omega_1,\omega_2)$ at the highest scale $n$ between frequency ranges were positive ($p \ll 0.0001$). 
We found moreover that the $\rho_{DCCA}$ values at the highest scale and $H_{amp}^\omega$ measured from the same neural data were highly correlated ($p \ll 0.0001$)
 and $H^\omega_{amp}$ values in distinct frequency ranges were highly correlated ($p\ll 0.0001$) (likewise for $\rho_DCCA(n)$). Finally we found that the $\rho_{DCCA}(n)$ and $H^\omega_{amp}$ values were not significantly correlated with $H_{raw}$ ($p > 0.05$). These results confirm the CPLAD hypothesis, in particular in agreement with predictions for 
 {\bf MU2}.
 
 Thus these findings provide strong evidence for criticality in the human brain by confirming the hypothesis of CPLAD. These results cannot be explained by the PF hypothesis; passive filtering of neural signals may generate a power-law spectrum but will not induce LRTC of narrowband amplitudes or DCCA correlations between narrowband amplitudes. Moreover, we note that the failure to detect avalanches in the experiments of \cite{bedard2006does,ribeiro2010spike} may be explained by criticality in the first meta-universality class where we have shown that even the largest avalanches are not discernible. Thus it is possible in these experiments that the systems under study were critical despite failure to demonstrate CPLAD directly.


We thank Daniel Bartz and Klaus-Robert M{\"u}ller for critical readings of the manuscript. Financial support; DAJB: DFG, GRK 1589/1 ÓSensory Computation in Neural SystemsÓ;
VVN partial support: BCCN, Berlin (BCCN 01GQ 1001C) and the National Research University, Higher School of Economics.

\bibliographystyle{ieeetr}
\bibliography{non_standard_fractional}

\begin{appendix}

\section{Estimators}
\label{sec:Estimators}

\subsection{Detrended Fluctuation Analysis}
\label{sec:Detrended Fluctuation Analysis}

Detrended Fluctuation Analysis (DFA) \cite{peng1994mosaic} is a methodology for the estimation of the Hurst exponent of a (possibly non-stationary) time-series.
Its advantage over covariance analysis or analysis of the power-spectrum are its robustness to trends contaminating the empirical time-series
and its desirable convergence properties \cite{weron2002estimating}.

The steps involved in DFA are as follows. First one forms the aggregate sum of the empirical 
 time-series $X(t)$:
 \begin{align*}
x(t) = \sum_{i=1}^t X(i)
 \end{align*}
 
(From now on whenever we refer to $x(t)$ we mean the time-series obtained from $X(t)$ by way of this operation.)
Analysis of the fluctuations in $X(t)$ may then be performed by measuring the variance of $x(t)$ in windows of varying size $n$ \emph{after} detrending, i.e., $x(t)$ is split into windows
of length $n$, $x_n^{(1)},\dots,x_n^{(j)},\dots,x_n^{(\lfloor N/n \rfloor)}$ and the average variance after detrending the data of $x(t)$ in these windows is formed; i.e. let $P_d$ be the operator which generates the mean-squares estimate of the polynomial fit of degree $d$, then the DFA coefficients or detrended variances of degree $d$ are:
\small
\begin{align*}
F_{DFA}^2(n) = \frac{1}{\lfloor N/n\rfloor \cdot n} \sum_{j}\left(x_n^{(j)}-P_d\left(x_n^{(j)}\right)\right)^\top\left(x_n^{(j)}-P_d\left(x_n^{(j)}\right)\right)
\end{align*}
 \normalsize
Crucially, in the limit of data the slope of $\text{log}(F^2_{DFA}(n))$ against $\text{log}(n)$ converges to $H$. Thus $X(t)$ is power-law correlated LRTC if and only if
the estimate of $H$, $\widehat{H}$, converges to a number greater than $0.5$ in the limit of data. We note here that there are numerous methods for the Hurst exponent; 
these include wavelet estimators \cite{abry1998long}, log-periodogram based methods \cite{giraitis2000adaptive}, among others \cite{robinson1995gaussian}. We chose DFA since it is standard in the physics and neuroimaging literature, 
and yields competitive estimates \cite{rea2009estimators}.

\subsection{Detrended Cross-Correlation Analysis}
\label{sec:Detrended Cross-Correlation Analysis}

In precise analogy to DFA, Podobnik and Stanley propose Detrended Cross-Correlation Analysis (DCCA) \cite{podobnik2008detrended}, an extension of DFA to two time-series, by considering the detrended covariance:
 \small
\begin{align*}
F_{DCCA}^2(n) = \frac{1}{\lfloor N/n \rfloor \cdot n} \sum_{j}\left((x_1)_n^{(j)}-P_d\left((x_1)_n^{(j)}\right)\right)^\top\left((x_2)_n^{(j)}-P_d\left((x_2)_n^{(j)}\right)\right)
\end{align*} 
 \normalsize
 
 Thus DCCA generalizes DFA in the sense that if $X_1=X_2$ then $F_{DCCA}^2(n) = F_{DFA}^2(n)$. 
 To simplify the fact that we will need to consider the DFA coefficients of $X_1$ and $X_2$ simultaneously, we will also refer to the DCCA coefficients as:
 \begin{align*}
 F_{X_1,X_2}^2(n) = F_{DCCA}^2(n)
 \end{align*}

Thus, DCCA quantifies the behaviour of the covariance between $X_1$ and $X_2$ over a range of time-scales given by $n$.
In analogy to the Pearson correlation coefficient, we may consider the \emph{detrended cross correlation coeffcient} \cite{zebende2011study}: 

\begin{align*}
\rho_{DCCA}(n,X_1,X_2) = \frac{F^2_{X_1,X_2}(n)}{\sqrt{F^2_{X_1,X_1}(n)}\sqrt{F^2_{X_2,X_2}(n)}} 
\end{align*}

$\rho_{DCCA}$ quantifies the correlation between $X_1$ and $X_2$ over a range of time-scales.
Note that while the machinery involved in the estimation of $\rho_{DCCA}$ are more complex than for Pearson correlation, 
both coefficients estimate the \emph{same} quantity for stationary time-series, not contaminated by trends. Thus $\rho_{DCCA}$
generalizes the Pearson correlation coefficient. Applicability to non-stationary time-series is particularly important for the neural data analysis.
Whenever the context allows for no ambiguity we abbreviate $\rho_{DCCA}(n,X_1,X_2)$ to $\rho_{DCCA}(n)$.

\section{Theory of the Avalanche Process}
\label{sec:Theory of the Avalanche Process}

We derive here properties which hold for the CPLAD model but not for the PF model. 
The theory focuses on the Hurst exponent of $X(t)$ which we denote $H_{raw}$, the Hurst exponent of $g_\omega(X(t))$, which we denote $H^\omega_{amp}$
and the DCCA correlation coefficient $\rho_{DCCA}(n, g_{\omega_1}(X(t)),g_{\omega_2}(X(t)))$.
Throughout we assume that all power-law distributions are cutoff at a lifetime $L_c$, that $q_s = q$ is fixed and deterministic for all $s$ and that $t$ is greater than $L_c$ to ensure stationarity.

First we consider the PF model.

\begin{Prop}
\label{prop:1}
For the PF model $\mathbb{E}\big(\rho_{DCCA}\big(n,g_{\omega_1}(X(t)),g_{\omega_2}(X(t))\big)\big) \rightarrow 0$
for $n \rightarrow \infty$ and $H^\omega_{amp} = 1/2$. 
\end{Prop}

\begin{proof}
For the passive filtering model, the right hand term of Equation~(3) of the main paper dominates.
That is:
$$Z(t) = \sum_{L_{s,i}\leq L'} h_{s,i} a_{s,i}\left(\frac{t-s}{L_{s,i}}\right)$$
Time points of $g_\omega(Z(t))$ greater than $L'$ apart are independent so $H^\omega_{amp} = 1/2$.
\newline Moreover   $\mathbb{E}\big(\rho_{DCCA}\big(n,g_{\omega_1}(X(t)),g_{\omega_2}(X(t))\big)\big) \rightarrow 0$ since
the numerator of the DCCA correlation coefficient is given by the DCCA coefficients $F^2_{X_1,X_2}(n)$. These behave
like the coefficients of white noise for large window sizes and therefore have expectation 0 in the limit.
\end{proof}

We now turn to the CPLAD model, to which all of the following propositions apply.
We first aim at understanding the scaling of the number of active large avalanches:

\begin{Prop}
\label{prop:2}
The number of active avalanches is asymptotically $\sim L_c^{2-\alpha}$ for large $L_c$ and $\alpha<2$ and of constant 
order for $\alpha>2$. 
\end{Prop}

\begin{proof}
The number of active avalanches at time $t$ is equal to:
\begin{align}
\label{eq:1}
& \sum_{t'=0}^{L_c} \#\{L_{t-t',i} | L_{t-t',i} > t' \}  \nonumber \\
&\sim \int_0^{L_c} q\left(  \int_{t'}^{L_c} L^{-\alpha} dL \right) dt'  \nonumber \\
&\sim (1-\alpha) (L_c)^{1-\alpha} -(1-\alpha) \int_0^{L_c} (t')^{1-\alpha} dt' \nonumber \\
&= (1-\alpha) (L_c)^{1-\alpha} - (1-\alpha)(2-\alpha)(L_c^{2-\alpha}-1) \nonumber \\
&\begin{cases} \sim L_c^{2-\alpha} \text{ for $\alpha<2$} \\  = \mathcal{O}(1) \text{ for $\alpha>2$} \end{cases}
\end{align}
\end{proof}

\begin{Prop}
\label{prop:3}
$X(t)$ is approximately Gaussian for large $L_c$ and for $\alpha<2$.
\end{Prop}

\begin{proof}
At any time point we have a number $k \sim L_c^{2-\alpha}$ (Proposition~\ref{prop:2}) of avalanches which commenced each at times
$s_i = 1,\dots, k$ with lengths $L_i$.
Then we may write:
\begin{align*}
X(t) = \sum_{i=1}^k L_i^\beta a\left(\frac{t-s_i}{L_i}\right)
\end{align*}
For $k$ large, which may be assumed by Proposition~\ref{prop:2}, after normalization and centering we apply the Lyapunov central limit condition \cite{billingsley1995probability} to:
\begin{align*}
\overline{X}(t) = \frac{1}{\sqrt{k} L_c^{\beta+(\alpha+1)/2}}\sum_{j=1}^k L_i^\beta \left( a\left(\frac{t-s_i}{L_i}\right) - \mu\right) 
\end{align*}
where $\mu$ is the mean value of $a(t)$ over the interval $[0,1]$.
Let $s_k^2 = k L_c^{2\beta+\alpha+1}$, then for the Lyapunov condition to apply we need to show, for some $\delta>0$, that:
\begin{align*}
\frac{1}{s_k^{2+\delta}} \sum_{j=1}^k \mathbb{E}\left|  L_i^\beta \left( a\left(\frac{t-s_i}{L_i}\right) - \mu\right) \right|^{2+\delta} \rightarrow 0
\end{align*}
Thus:
\begin{align*}
&\frac{1}{s_k^{2+\delta}} \sum_{j=1}^r \mathbb{E}\left|  L_i^\beta \left( a\left(\frac{t-s_i}{L_i}\right) - \mu\right) \right|^{2+\delta} \\
&\sim \frac{1}{k^{1+\delta/2} L_c^{(\beta+\alpha/2+1/2)(2+\delta)}} \sum_{i=1}^k  L_i^{(2+\delta)\beta} \\
&\sim \frac{1}{k^{1+\delta/2} L_c^{(\beta+\alpha/2+1/2)(2+\delta)}} k L_c^{\beta(2+\delta)+\alpha+1} \\
&\sim \frac{L_c^{\alpha+1}}{k^{\delta/2} L_c^{(\alpha/2+1/2)(2+\delta)} } \\
& \leq  \frac{1}{k^{\delta/2} L_c^{\delta/2} } \longrightarrow 0
\end{align*}

This completes the proof for univariate Gaussianity. The proof that convergence is to a multivariate Gaussian process is a simple extension
via the Cramer-Wold device \cite{billingsley1995probability}.

\end{proof}

\begin{Prop}
\label{prop:4}
For $\alpha<2$, $\mathbb{E}\big(\rho_{DCCA}\big(n,g_{\omega_1}(X(t)),g_{\omega_2}(X(t))\big)\big) \rightarrow 0$
for $n \rightarrow \infty$ and $H^\omega_{amp} = 1/2$.
\end{Prop}

\begin{proof}
By Proposition~\ref{prop:3}, $X(t)$ is a Gaussian process, which is uniquely defined by its first two moments.
Over long time-scales, therefore, $X(t)$ behaves like fractional Brownian noise \cite{taqqu1975weak}.
One approach to generating fractional Brownian noise is by long-range filtering of a Gaussian white noise \cite{amblard2013basic}.
This brings us exactly into the situation of Proposition~\ref{prop:1}.
\end{proof}

\begin{Prop}
\label{prop:5}
The probability that avalanches of length $L'$ occur simultaneously vanishes for large $L'$ when $\alpha>2$. 
\end{Prop}

\begin{proof}
In analogy to Proposition~\ref{prop:2} equal to $L'$ is the time spent in large avalanches divided by the total recorded time:

\begin{align*}
\frac{1}{T}\sum_{L_i>L'} L_i &\sim q \mathbb{E}(\mathbbm{1}_{(L',\infty)}(L)) \\
&\sim \int^{L_c}_{L'} L^{1-\alpha} d L \\
&\sim L'^{2-\alpha}
\end{align*}

Thus the probability that two avalanches of length greater than $L'$ occur $\rightarrow 0$.
\end{proof}

\begin{Prop}
\label{prop:6}
The Hurst exponent $H_{raw}$ of $X(t)$ satisfies the approximate relation:
\begin{align}
\label{eq:2}
H_{raw} &= 
\begin{cases} \beta -\frac{\alpha}{2} + 2 \hspace{3mm} &\text{ if $\tau< 3$}  \\
1/2 &\text{ if $\tau>3$}
\end{cases}
\end{align}
\end{Prop}

\begin{proof}
Following \cite{jensen19891} we approximate avalanches by a box function: 
\begin{align*}
a(t) = \begin{cases} 1 \text{ for }t \in [0,1] \\ 0 \text{ otherwise} 
\end{cases}
\end{align*}

The authors show that in this case the autocorrelation function $r(t) = \mathbb{E}(X(s)X(s+t)) - \mathbb{E}(X(s))^2$ satisfies:
\begin{align*}
r(t) &\propto \int_{|t|}^\infty (L-|t|)\int_0^\infty L^{-\alpha} L^{2\beta} dL \\
&\sim |t|^{2\beta-\alpha + 2}
\end{align*}
Using the fact that if the autocorrelation function scales as $r(t) \sim {t^{-\gamma}}$, then the 
Hurst exponent and $\gamma$ are related as $\gamma = 2-2H$ \cite{kantelhardt2001detecting},
so that:
\begin{align}
\label{eq:3}
H_{raw} &= \beta-\frac{\alpha}{2}+2
\end{align}
\end{proof}

We adopt the convention that by $g_\omega(a)(t)$ we refer to the $t^\text{th}$ time-point of $g_\omega(a(t))$.
At criticality we have the relationship:
$$g_\omega(L_i^\beta a(t/L_i)) \sim L_i^{\beta'} g_\omega(a)(t/L_i)$$
This scaling relationship is illustrated in Figure~\ref{fig:filtered_illustration}. 
The left hand panel illustrates avalanches whose heights scale according to $L^\beta$, with $\beta = 1$.
The right hand panel displays $g_\omega(a(t/L))$. One sees that the heights of these filtered avalanches
scale more slowly than the original avalanches.

We require an additional known relation between critical exponents.
Let $S$ be the size of an avalanche, i.e.~the total activity occurring over the course of the avalanche: $S = \int_0^L L^\beta a(t/L) dt$.
Then, at criticality:
\begin{align*}
S \sim S^{-\tau}
\end{align*}
Then we have the following relation \cite{sethna2001crackling} between critical exponents:
\begin{align}
\label{eq:4}
\frac{\alpha-1}{\tau-1} = \beta+1
\end{align}
If we define $\alpha' = \alpha$ then we may also define $\tau'$ by Equation~\eqref{eq:4} so that:
\begin{align}
\label{eq:5}
\frac{\alpha'-1}{\tau'-1} = \beta'+1
\end{align}
The size distribution is important for our subsequent analyses because whether the asymptotic properties of the process $X(t)$
are dominated by the large avalanches depends on whether the size distributions have infinite variance or not.

\begin{figure}
\begin{center}
\includegraphics[width=140mm]{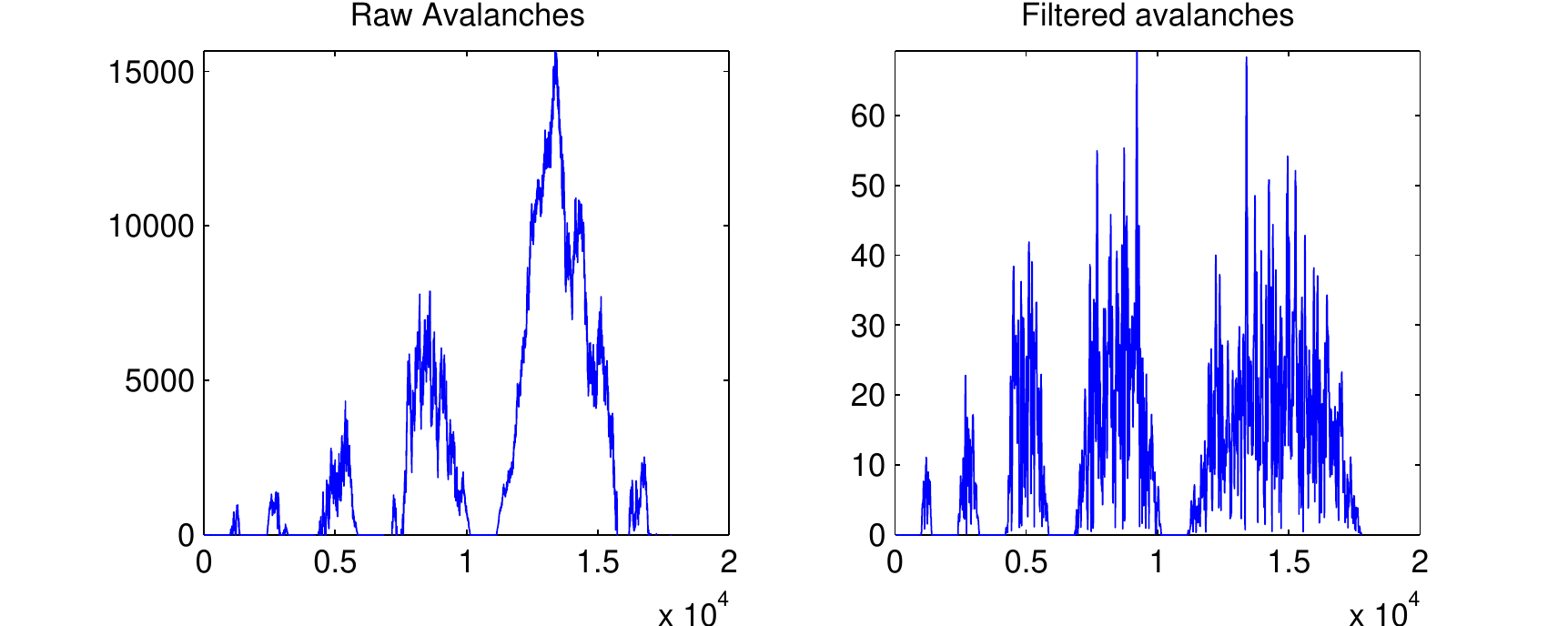}
\end{center}
\caption[Scaling of Filtered Avalanches]{The figure illustrates the difference in scaling between 
the raw avalanches and their filtered amplitudes, when $\beta = 1$.
The left hand panel displays avalanches $L^\beta a(t/L)$ with log-spaced lifetimes $L$ between 400 and 7000.
The right hand panel displays the narrowband amplitudes of these avalanches $L^{\beta'} g_\omega(a)(t/L)$.
Since $\beta > \beta'$ we see that the narrowband amplitudes scale less steeply than the raw avalanches. See Section~\ref{sec:Theory of the Avalanche Process}.}
\label{fig:filtered_illustration}
\end{figure}

\newpage

\begin{Prop}
\label{prop:7}
 $\beta' = \beta/2$. 
 \end{Prop}
 
 \begin{proof}
 We have that $\mathcal{P}[\omega] \sim \omega^{-\beta-1}$. This implies that
 the standard deviation of $f_{L\omega, L\Delta \omega}(a(t))$ scales according to $L^{-\beta/2}$.
 This is because: 
 \begin{align*}
 \text{var}(f_{\omega, \Delta \omega}(a(t))) &\sim \int_{\omega-\Delta \omega}^{\omega+\Delta \omega} \mathcal{P}[\omega'] d\omega' \\
 &\sim \frac{1}{\omega^{\beta+1}}
 \end{align*}
 And therefore:
  \begin{align*}
 \text{var}(f_{L \omega, L \Delta \omega}(a(t))) &\sim \int_{L \omega- L \Delta \omega}^{L \omega+L \Delta \omega} \mathcal{P}[\omega'] d\omega' \\
 &\sim  \int_{L \omega- L \Delta \omega}^{L \omega+L \Delta \omega}\frac{1}{\omega'^{\beta+1}} d\omega' \\
 &\sim   \int_{\omega- \Delta \omega}^{\omega+ \Delta \omega}\frac{1}{(L\omega')^{\beta+1}} L d\omega' \\
 &\sim \frac{1}{L^\beta} \text{var}(f_{\omega, \Delta \omega}(a(t)))
 \end{align*}
 Therefore: $$\text{std}(f_{L\omega, L\Delta \omega}(a(t)))) \sim L^{-\beta/2}$$
 and: 
 \begin{align*}
 f_{\omega,\Delta \omega}( L_i^\beta a_{i}(t/L_i)) &\sim L_i^\beta f_{L_i \omega,L_i \Delta \omega}(a)(t/L_i) \\
 &\sim L_i^{\beta/2} f_{\omega,\Delta \omega}(a) (t/L)
 \end{align*}
 Therefore by linearity of the Hilbert transform:
 $$ g_\omega(  L_i^\beta a_{i}(t/L_i)) \sim L_i^{\beta/2} g_{\omega}(a) (t/L_i)$$
 \end{proof}

\begin{Prop}
\label{prop:8}
The Hurst exponent $H^\omega_{amp}$ of $g_\omega(X(t))$ satisfies the approximate relation:
\begin{align}
\label{eq:6}
H_{amp} &= 
\begin{cases} \frac{\beta}{2} -\frac{\alpha}{2} + 2 \hspace{3mm} &\text{ if $\alpha \leq \beta+3$, $\alpha>2$} \\
1/2 &\text{ otherwise }
\end{cases}
\end{align}
\end{Prop}

\begin{proof}
$X(t)$ may be divided into a sum of long avalanches $L_{s,i}>L'$ which do not overlap and short avalanches.
\begin{align}
\label{eq:7}
X(t) =  \sum_{L_{s,i}>L'} L_{s,i}^\beta a_{s,i}\left(\frac{t-s}{L_{s,i}}\right) +  \sum_{L_{s,i}\leq L'} L_{s,i}^\beta a_{s,i}\left(\frac{t-s}{L_{s,i}}\right) 
\end{align}
Therefore by Proposition~\ref{prop:7}:
\begin{align}
\label{eq:8}
f_\omega(X(t)) =  \sum_{L_{s,i}>L'} L_{s,i}^{\beta/2} f_\omega \left(a_{s,i}\left(\frac{t-s}{L_{s,i}}\right)\right) +  
\sum_{L_{s,i}\leq L'} L_{s,i}^{\beta/2} f_\omega\left(a_{s,i}\left(\frac{t-s}{L_{s,i}}\right)\right) 
\end{align}
The left hand term dominates whenever its variance is unbounded for large $L_c$. This 
happens if $\tau'<3$ which translates to $\alpha<\beta+3$ by Equation~\eqref{eq:5}.
Therefore the right hand term may be neglected and since the avalanches of the left hand term are 
separated in time, the envelope operator may be pulled under the sum so that:
\begin{align}
\label{eq:9}
g_\omega(X(t)) \sim  \sum_{L_{s,i}>L} L_{s,i}^{\beta/2} g_\omega \left(a_{s,i}\left(\frac{t-s}{L_{s,i}}\right)\right)
\end{align}
The same proof as for Proposition~\ref{prop:6} may then be applied to deriving $H_{amp} = \frac{\beta}{2} -\frac{\alpha}{2}$.
On the other hand if $\tau'>3$ then the right hand term dominates, for large $L'$ and has time-scale bounded by $L'$---thus in this
case $H_{amp} = 1/2$.
\end{proof}

We now study the DCCA correlation coefficients between amplitude envelopes.

\begin{Prop}
\label{prop:9}
Assuming that $\int_0^1 g_\omega(a(t)) dt$ converges, for $\alpha<\beta+3$ and $\alpha>2$, $\rho_{DCCA}(n,g_{\omega_1}(X(t)),g_{\omega_2}(X(t)) \rightarrow 1$
for $n \rightarrow \infty$. 
\end{Prop}

\begin{proof}
For large $n$, each DCCA window is longer than the largest avalanche. 
For $\alpha<\beta+3$ the variance in the left hand term of Equation~\eqref{eq:8} may be assumed to be non-overlapping since we assume $\alpha>2$ (Propostion~\ref{prop:5}) we may neglect small avalanches and assume
that a window contains only one avalanche at $t'$ with length $L$. Then:
$$\sum_{i=1}^{t} g_{\omega_j}(X)(i) = \begin{cases} 0  &\text{ if $t<t'$} \\
 L^{\beta/2} \int_{0}^{L} {g_{\omega_j}}(a(t/L)) &\text{ if $t>t'+L$}
 \end{cases}$$
 But $\int_{0}^{L} {g_{\omega_1}}(a(t/L)) \sim \int_{0}^{L} {g_{\omega_2}}(a(t/L))$ up to a constant factor for large avalanches and assuming 
 the integral converges. Thus the correlation tends to 1.
\end{proof}

\begin{Prop}
\label{prop:10}
Assuming that $\int_0^1 a(t) dt$ converges, for $\tau<3$ and $\alpha>2$, \newline $\rho_{DCCA}(n,g_{\omega_1}(X(t)),g_{\omega_2}(X(t)) \rightarrow 1$
for $n \rightarrow \infty$ and $\omega_i \rightarrow 0$. 
\end{Prop}

\begin{proof}
For $\omega_1,\omega_2 > L_c$ then we have that $g_{\omega_j}(L^\beta a(t/L)) \sim L^\beta g_{\omega_j}(L^\beta a(t/L))$. This
is because, $a(t/L)$ may be approximated by a delta function relative to these frequencies, which has white spectrum. Then 
the same proof techniques of Propositions~\ref{prop:8} and~\ref{prop:9} may be applied but with Equation~\eqref{eq:8} replaced by:

\begin{align}
\label{eq:10}
f_\omega(X(t)) =  \sum_{L_{s,i}>L} L_{s,i}^{\beta} f_\omega \left(a_{s,i}\left(\frac{t-s}{L_{s,i}}\right)\right) +  
\sum_{L_{s,i}\leq L} L_{s,i}^{\beta} f_\omega\left(a_{s,i}\left(\frac{t-s}{L_{s,i}}\right)\right) 
\end{align}
For $\tau>3$, the left hand term dominates, since it has unbounded variance in $L_c$, which completes the proof.
\end{proof}

\begin{Prop}
\label{prop:11}
For $\alpha<2$, $\mathbb{E}\big(\rho_{DCCA}\big(n,g_{\omega_1}(X(t)),g_{\omega_2}(X(t))\big)\big) \rightarrow 0$
for $n \rightarrow \infty$, $H^\omega_{amp} = 1/2$ and $H^\omega_{raw}$.
\end{Prop}

\begin{proof}
The proof is identical to the proofs of Propositions~\ref{prop:1} and~\ref{prop:4}, using the results of Proposition~\ref{prop:6}.
\end{proof}

\subsection{Summary of Results}
\label{sec:Summary of Results}

We term each region of parameters a \emph{meta-universality class} which 
gives qualitatively constant behaviour as quantified by Hurst exponents and DCCA correlation coefficients. There are four meta-universality classes {\bf MU1--MU4}:

\begin{itemize}
\item[\bf MU1:] 
$\alpha<2$:
\begin{align*}
&H_{raw} > 0.5  \\
 &H_{amp}^\omega = 0.5 \\
  &\rho_{DCCA}(n, g_{\omega_1}(X(t)), g_{\omega_2}(X(t))) \rightarrow 0
\end{align*}
\item[\bf MU2:] $\alpha>2$ and $\alpha \leq \beta+3$
\begin{align*}
&H_{raw} > 0.5 \\ 
&H_{amp}^\omega > 0.5 \\
&\rho_{DCCA}(n, g_{\omega_1}(X(t)), g_{\omega_2}(X(t))) \rightarrow 1
\end{align*}
\item[\bf MU3:] $\alpha >2$ and $\tau<3$
\begin{align*}
& H_{raw} > 0.5 \\
& H_{amp}^\omega = 0.5 \\
& \rho_{DCCA}(n, g_{\omega_1}(X(t)), g_{\omega_2}(X(t))) \rightarrow 1 \text{ as $\omega_i \rightarrow 0$}
\end{align*}
\item [\bf MU4:] $\alpha >2$ and $\tau>3$
\begin{align*}
& H_{raw} = 0.5 \\
& H_{amp}^\omega = 0.5 \\
& \rho_{DCCA}(n, g_{\omega_1}(X(t)), g_{\omega_2}(X(t))) \rightarrow 0
\end{align*}

\end{itemize}

\section{Simulations}
\label{sec:Simulations}
In all simulations we model the average avalanche shape as a quadratic function: $b(t) = -4(t-1/2)^2+1$. 
Moreover  we set the variance swell identically so that $b(t) = c(t)$. 
For the noise component we take the implementation of \cite{kasdin1995discrete}.
For the power-law cutoff sampling, we perform a density transformation of the uniform distribution.
(In MATLAB {\tt x = rand(1,T).*(1-$L_c$)+$L_c$; x = 1./(x.\^~(1./($\alpha$-1)))})

\subsection{Simulation 1 (from main text)}
\label{sec:Simulation 1 (from main text)}

We describe here how we obtain the results of Figure~2 of the main text.

Lower panel: we filter white Noise with $T = 40000$ (using the method of \cite{kasdin1995discrete}) to yield a process with spectrum scaling according to $1/ \omega^{1.9}$. We then measure DCCA correlations between
amplitudes in the frequency ranges [0.68,0.72] and [0.78,0.82] of half the sampling frequency.
 with $n$ log-spaced between 20 and 9000 and Hurst exponents of the same amplitudes using DFA with window sizes between 1000 and 9000.
 
 Top two panels: we generate two processes assuming CPLAD applied to Equation (1) of the main paper, generating avalanches as described above. We set $q = 5$, $T = 40000$ $\beta = 1$ and $\alpha = 2.5, 1.5$ (top, middle resp.).
 We then apply the same analysis as in the lower panel.

\subsection{Simulation 2}
\label{sec:Simulation 2}

The aim of this simulation is to verify that long-range dependent Gaussian processes satisfy $H^\omega_{amp} = 0.5$ and $\rho_{DCCA} = 0$, where
the time series are generated as filtered Gaussian white noise processes. Thus this simulation tests our predictions for the PF model.

To this end we simulate 100 long-range dependent Gaussian processes ($H>0.5$) (using the method of \cite{kasdin1995discrete}) of length 15000 time points. In each case the data are then filtered forward and backwards in 
two separate frequency bands (between 0.39 and 0.41 of the sampling frequency and 0.29 and 0.31 of the sampling frequency)
with butterworth filters of order $n$. We then measure DCCA correlations between the Hilbert transforms of these signals and
measure their Hurst exponents with DFA, in both cases using window lengths between $10^3$ and $10^4$.

We repeat this setup for $H=0.9,1.4$ and $n=2,4$.

The results are displayed in Figure~\ref{fig:gaussian_processes} and show that, up to small sample effects, we expect $H_{amp}=0.5$ and zero cross correlations $\rho_{DCCA} = 0$ between frequency bands.

\begin{figure}
\begin{center}
$\begin{array}{c}
\includegraphics[width=130mm,clip=true,trim= 0mm 0mm 0mm 0mm]{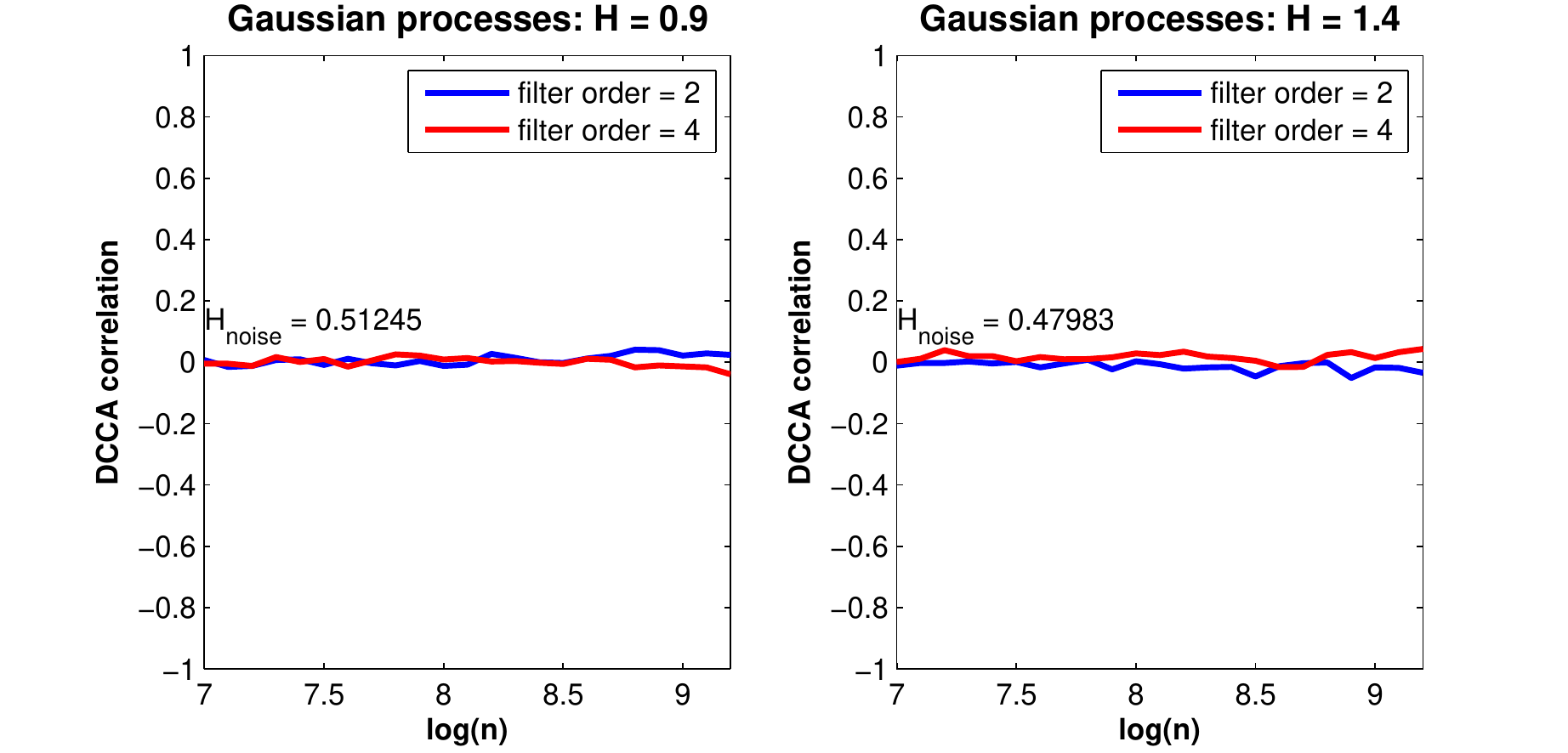}  
\end{array}$
\end{center}
\caption{The figure displays the results of the long-range analysis for a Gaussian process model.}
\label{fig:gaussian_processes}
\end{figure}

\subsection{Simulation 3 (from main text)}
\label{sec:Simulation 3 (from main text)}

In the first simulation, for each pair of exponents in the ranges $\alpha = 1.5,\dots, 6.5$ and $\beta = 0.25,0.5,\dots,3$,
we generate a sample path $X(t)$ of length $T = 300,000$, with a cutoff at $L_c = 100,000$, and number of superpositions $q=5$. The first $100,000$ time points are discarded, to ensure stationarity. 
We then design Butterworth filters of order 2 between $0.29$ and $0.31$ and between $0.39$ and $0.41$ of the sampling frequency.
The data from $X(t)$ are then filtered forwards and backwards (yielding effective filter order of 4) and the amplitude envelopes are calculated to yield $Y_1(t)$ and $Y_2(t)$.
We measure the Hurst exponent of $\mathcal{E}(X_{\omega_1}(t))$, using DFA, and the DCCA correlation coefficients between $\mathcal{E}(X_{\omega_1}(t))$ and $\mathcal{E}(X_{\omega_2}(t))$,
setting $n$ to log spaced values between $100$ and $200000$.
This setup is repeated 100 times, and the results of the simulations are averaged.

\subsection{Simulation 4}
\label{sec:Simulation 4}

In the second simulation we check the prediction of $H_{amp}$ and $H_{raw}$, setting 
set $\alpha = 2.5$, $q=1$, $L_c = 10^6$, with a burn in time of $10^5$ time points (less time is required for convergence for this value of $\alpha$), setting $n$ to log spaced
values between $7000$ and $60000$ for the estimation of $H_{amp}$ and between $100$ and $7000$ for the estimation of $H_{raw}$ (according to where scaling regions were observed).
The results are displayed in Figure~\ref{fig:cutoff_corr_theory},

The quality of the $H_{raw}$ estimate is greater for small $\beta$, whereas the quality of the $H_{amp}$ estimate is greater for larger $\beta$. This discrepancy may be explained as follows:
since $X(t)$ has longer tails than $\mathcal{E}(X_{\omega}(t))$, the convergence of its moments is slow for large $\beta$, thus the quality of the estimate decreases for large $\beta$.
On the other hand, the estimate of $H_{raw}$ requires a linear approximation to the non-linear transform given by the amplitude of the analytic signal: this 
approximation increases in quality for larger $\beta$.

\begin{figure}
\begin{center}
\includegraphics[width=80mm,clip=true,trim= 0mm 0mm 0mm 0mm]{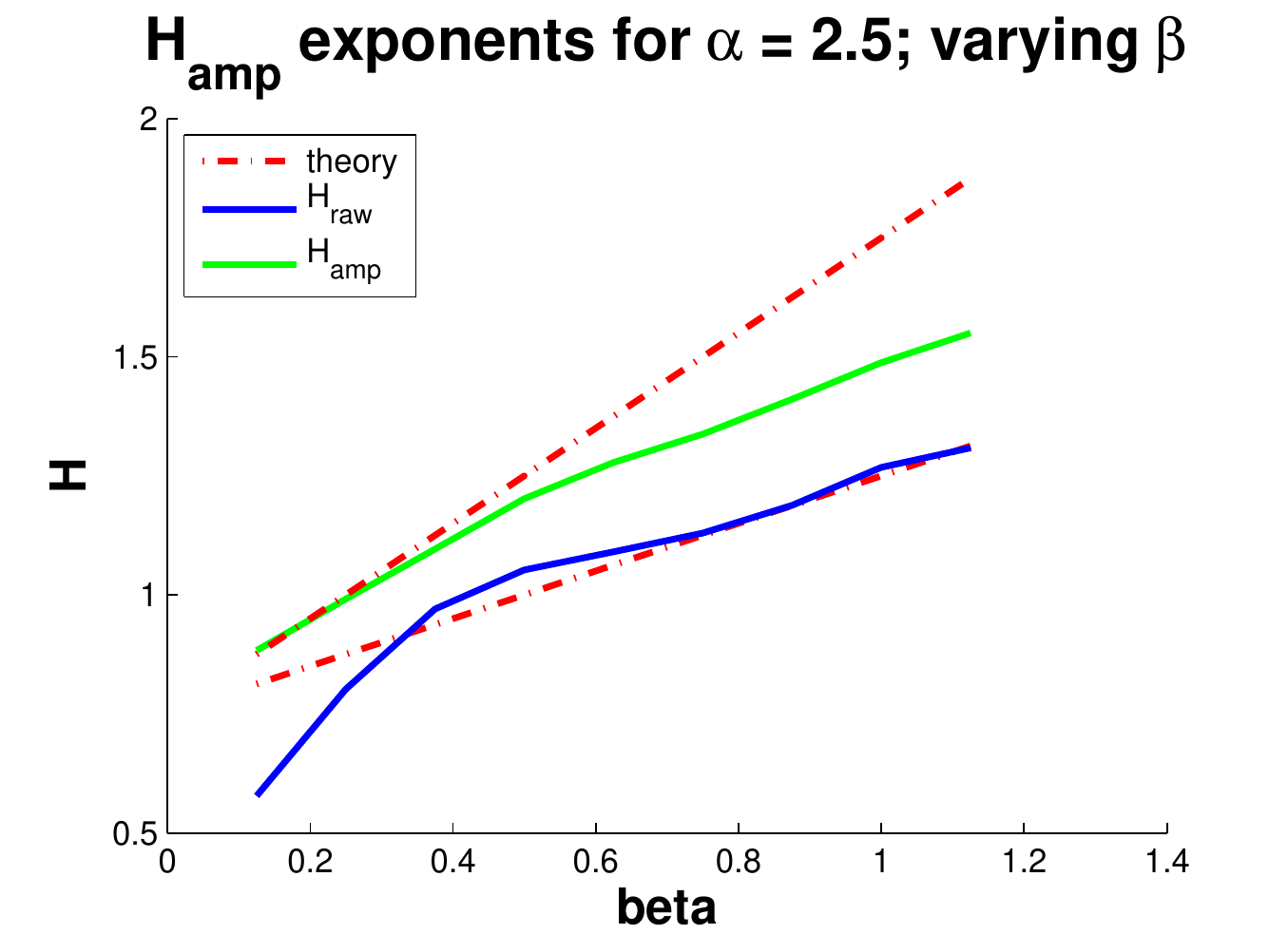}  
\end{center}
\caption{Scaling of $H_{amp}$ and $H_{raw}$ (second simulation) compared to theory. The theoretical estimate is $H_{amp} \sim 2-\alpha/2 +\beta'$ for $\beta'=1$.}
\label{fig:cutoff_corr_theory}
\end{figure}

\subsection{Simulation 5}
\label{sec:Simulation 5}

The aim of this simulation is to verify the theory of Proposition~\ref{prop:7} checking the heights of the filtered avalanches $a_\omega(t)$. We set $\beta = 0.25$, $L = 2^{10},\dots,2^{13}$ and for each $L$ considered, we simulate $10^4$ avalanches of this length, and calculate the mean avalanche profile. We then log-regress the height of these profiles against $\text{log}(L)$.
The results are displayed in Figure~\ref{fig:filtered_avalanches}. Close agreement is observed between the theoretical estimate $\beta' = \beta/2$
and the simulated results; the prediction improves for higher $\beta$.

\begin{figure}
\begin{center}
\includegraphics[width=80mm,clip=true,trim= 0mm 0mm 0mm 0mm]{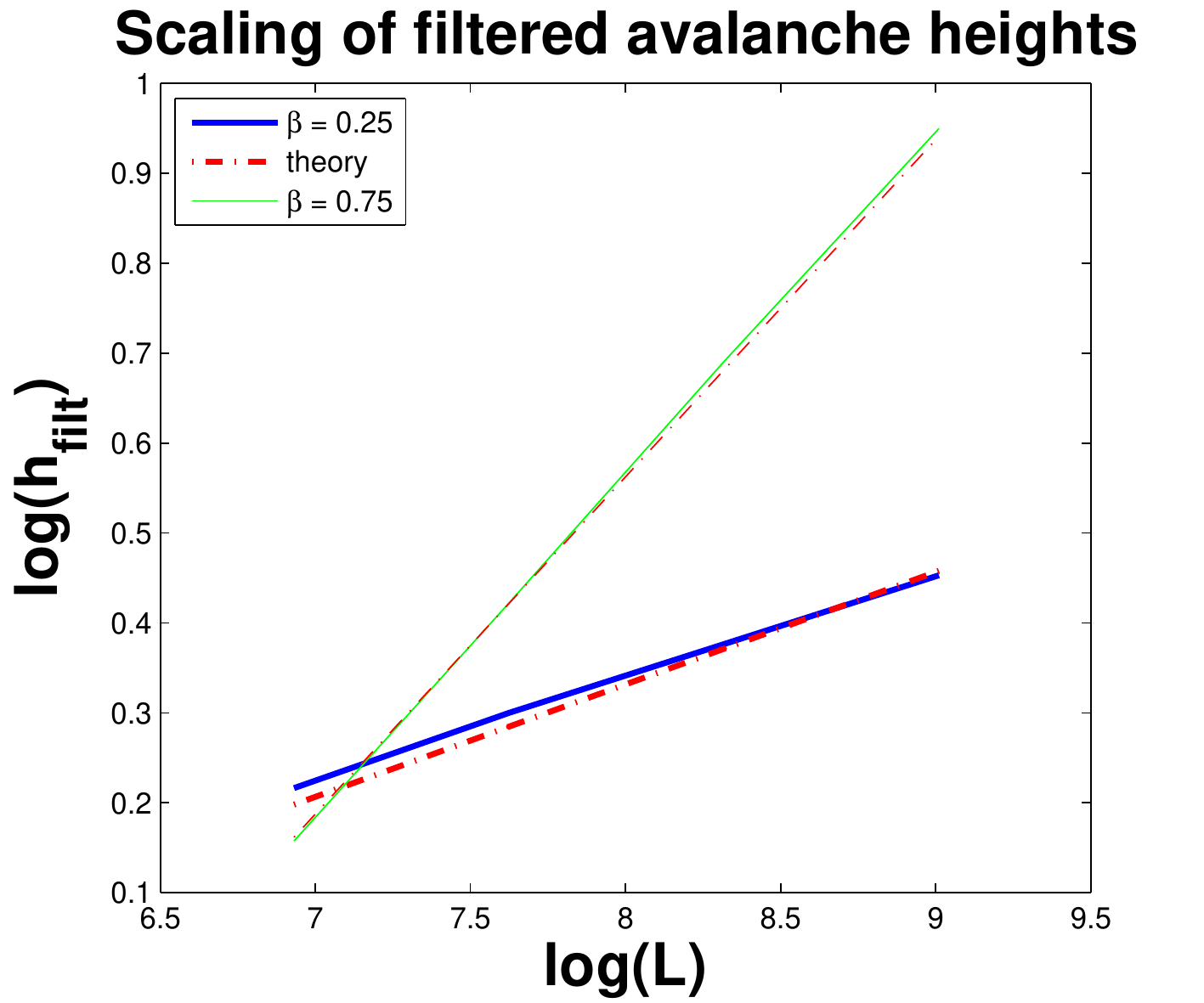}  
\end{center}
\caption{Comparison of the estimate $\beta/2$ for the scaling of heights of filtered avalanches vs. simulation. }
\label{fig:filtered_avalanches}
\end{figure}

\section{Data Analysis}
\label{sec:Data Analysis}

Seven subjects participated in the study (1 female). The experimental protocol was approved by the Institutional Review Board of the Charit{\'e} Medical University, Berlin.
EEG recordings were obtained at rest with subjects seated comfortably in a chair with their eyes open. Recordings were made of three sessions, each 5 minutes long so that each data set comprises roughly 15 minutes of data. EEG data were recorded with 96 Ag/AgCl electrodes, using BrainAmp amplifiers and BrainVision Recorder software (Brain Products GmbH, Munich, Germany). The signals were recorded in the 0.016--250 Hz frequency range at a 1000Hz sampling frequency and subsequently 
subsampled to 200Hz. 

The data analytic steps taken on the EEG data were as follows.
Outlier channels were rejected after visual inspection for abrupt shifts in voltage and poor signal quality. The data were then re-referenced according to the common average.
Spatial filters were computed on the data using Spatio-Spectral Decomposition (SSD) \cite{nikulin2011novel}, in order to extract components with pronounced alpha oscillations. Spatial filters with poor signal quality or topographies were rejected. We 
then restricted our analysis to those filters displaying a peak in the alpha range; this step ensured a high signal quality with low levels of artifactual activity.
The fact that the spatial filters yield clear oscillatory signals ensured that the neuronal processes in the adjacent frequency ranges similarly originated from cortical areas relating to neuronal rather
than artifactual activity. For DFA and DCCA estimation we set $n$ to log-spaced values between $1000$ and $25000$.

\subsection{Further Analyses}
\label{sec:Further Analyses}

In this section we repeat the analysis of real data reported in the main paper, but with a differing choice of spatial filter. Instead of taking SSD filters, which use the signal-to-noise ration in the alpha range 
to obtain filters corresponding to neuronal activity, we take Laplacian spatial filters \cite{srinivasan1996spatial}, which reduce redundancy between the data recorded at each electrode.
All other parameters remain fixed. In this way, we show that the results we obtain do not depend on the alpha rhythm.

The results are displayed in Figure~\ref{fig:real_data_prl_laplace} and confirm that the results obtained on real data using the SSD filters calibrated to the alpha rhythm may be obtained 
using filters which are independent of the alpha rhythm.

To further check that results do not depend on the detrending degree $d$, we repeat this analysis with $d=2$. 
The results are displayed in Figure~\ref{fig:real_data_prl_laplace_degree_2} and show qualitatively identical trends; most notable is that the results
are less noisy for higher degree detrending, which we conjecture is due to a more thorough removal of trends for quadratic detrending.

We also compute power-spectra of the first SSD component for each of the 7 subjects, using Welch's method (Hanning window of length 2 seconds, with 0 overlap between windows and calculated at 256 points).
The spectra are displayed in Figure~\ref{fig:psd}; each spectrum clearly displays the presence of alpha/ mu range oscillations at around 10Hz.

\begin{figure}
\begin{center}
\includegraphics[width = 150mm]{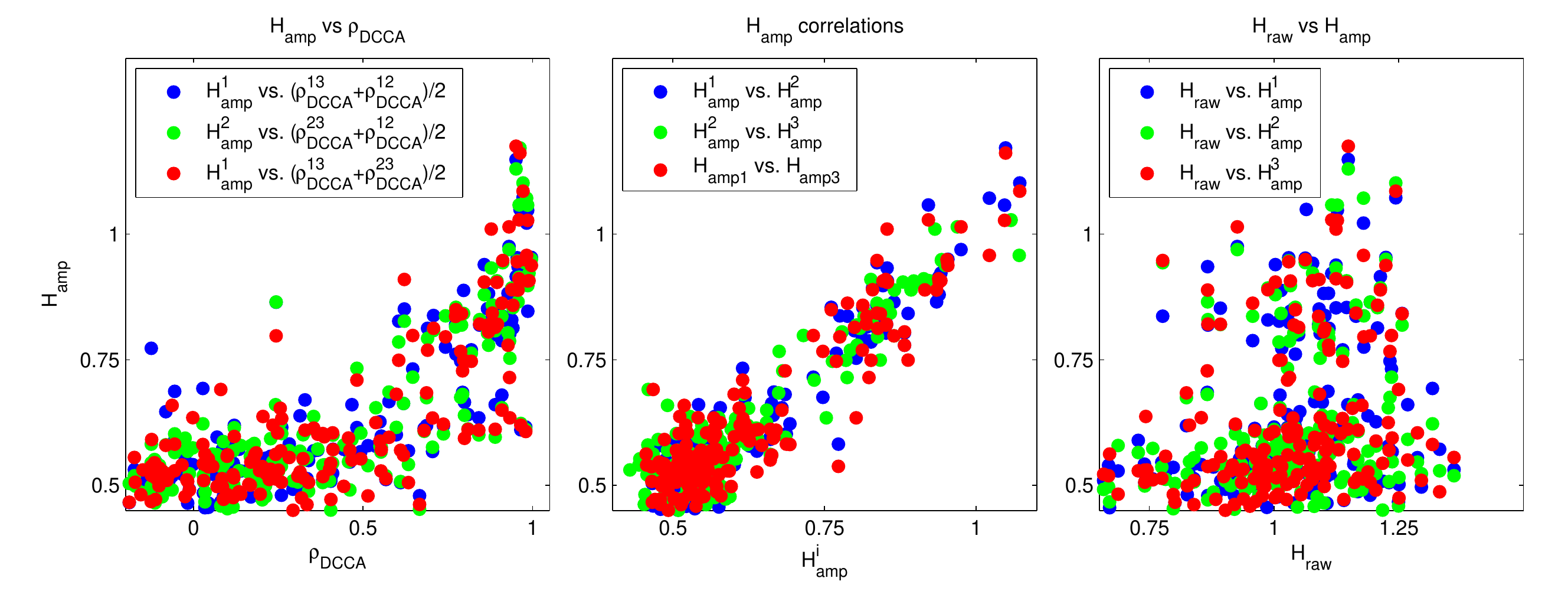}
\end{center}
\caption{The figure displays results in an identical analysis as carried out for Figure~4 of the main paper, with 
the only difference being that we use Laplacian rather then SSD spatial filters.}
\label{fig:real_data_prl_laplace}
\end{figure}

\begin{figure}
\begin{center}
\includegraphics[width = 130mm]{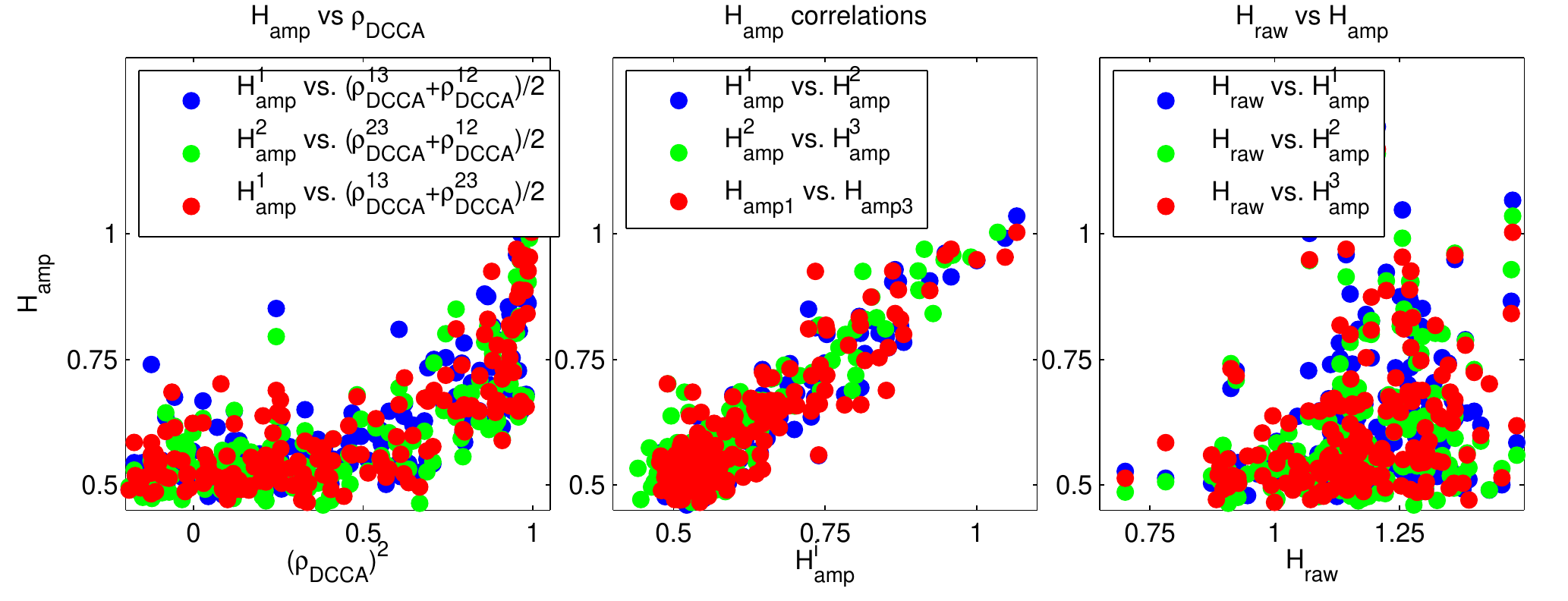}
\end{center}
\caption{The figure displays results in an identical analysis as carried out for Figure~4 of the main paper, with 
the only difference being that we use Laplacian rather then SSD spatial filters.}
\label{fig:real_data_prl_laplace_degree_2}
\end{figure}

\begin{figure}
\begin{center}
\includegraphics[width = 150mm]{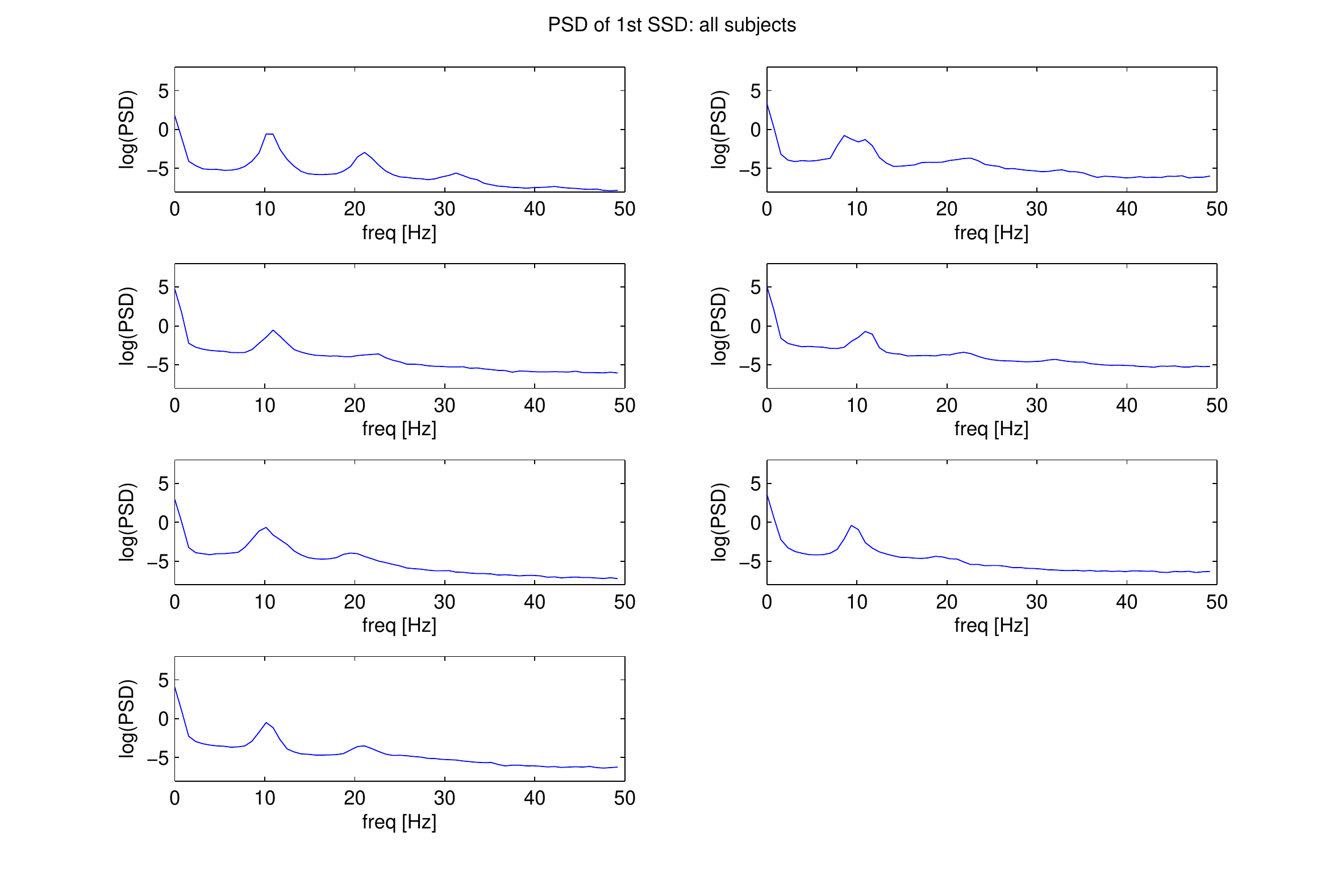}
\end{center}
\caption{The figure displays power-spectra from each of the 7 subjects considered in the study. 
In case the power-spectrum is estimated on the first SSD component using Welch's method.}
\label{fig:psd}
\end{figure}

\end{appendix}

\end{document}